%% file: main.tex
\documentclass[11pt, letterpaper]{article}
\usepackage{amsmath,amssymb,amsthm} 
\usepackage{fullpage}
\usepackage{hyperref}
\usepackage{xspace}
\usepackage{xcolor}
\usepackage{stmaryrd}
\usepackage[ruled,vlined,linesnumbered]{algorithm2e}
\usepackage{booktabs}
\usepackage{authblk}

\newtheorem{lemma}{Lemma}
\newtheorem{definition}{Definition}
\newtheorem{theorem}{Theorem}

\newcommand{\feed}[1]{\ensuremath{\mathsf{#1Binary}}}
\newcommand{\fullfeed}[1]{\ensuremath{\mathsf{#1Full}}}
\newcommand{\gfeed}{\mathsf{Result}\xspace}

\newcommand{\abfeed}{(\alpha,\beta)\text{-}\mathsf{Feed}\xspace}

\renewcommand{\Pr}[1]{\mathbb{P}\left[\,#1\,\right]}

\newcommand{\ie}{{\it i.e.,}\xspace}

\newcommand{\dk}[1]{{#1}}
\newcommand{\dpj}[1]{{#1}}
\newcommand{\gt}{Group Testing\xspace}

\newcommand{\cO}{O}
\newcommand{\ep}{\varepsilon}
\newcommand{\polylog}{{\rm \ polylog\ }}

\newcommand{\cT}{{\mathcal{T}}}
\newcommand{\cS}{{\mathcal{S}}}
\newcommand{\cR}{{\mathcal{R}}}
\newcommand{\cF}{{\mathcal{F}}}
\newcommand{\fieldF}{{\mathbb{F}}}
\newcommand{\symdiff}{\;\triangle \;}
\newcommand{\remove}[1]{}

\begin{document}

\title{Efficient Deterministic Quantitative Group Testing for Precise Information Retrieval\thanks{D. Pajak was supported by the National Science Centre, Poland—Grant Number 2019/33/B/ST6/02988.}}

	\author[1]{Dariusz R. Kowalski}
	\author[2]{Dominik Pajak}
	\affil[1]{School of Computer and Cyber Sciences, Augusta University, USA, dkowalski@augusta.edu}
	\affil[2]{Wrocław University of Science and Technology, Poland, dominik.pajak@pwr.edu.pl}
	\date{}
\maketitle

\begin{abstract}
 \input{abstract}
\end{abstract}


\input{intro}

\input{examples}
\input{model}


\input{constructive}
\input{applications}
\input{selectors}

\input{strong}

\input{size-feedback-lower}

\input{largeK}

\input{future}

\bibliographystyle{abbrv}
\bibliography{biblio} 





\end{document}

%% file: abstract.tex
The Quantitative Group Testing (QGT) 
is about learning a (hidden) subset $K$ of some 
large domain $N$
using 
a
sequence of queries, where a 
result of
 a query provides 
 information about the size of the intersection of the query with the unknown subset $K$. 
Almost all previous work focused on randomized algorithms minimizing the number of queries; 
however, in case of large domains $N$, randomization
may result in
a significant deviation from the expected precision. Others assumed unlimited computational power (existential results) or adaptiveness of queries.
In this work we propose efficient 
{\em non-adaptive deterministic} QGT algorithms for constructing queries and deconstructing a hidden set $K$ from the results of the queries, without using randomization, adaptiveness or unlimited computational power. The efficiency is three-fold. First, in terms of almost-optimal number of queries -- we improve it by factor nearly $|K|$ comparing to previous constructive results.
Second, our algorithms 
construct the queries and reconstruct set $K$
in polynomial time. Third, they work for {\em any} hidden set $K$, as well as multi-sets, and even if the results of the queries are capped at $\sqrt{|K|}$. We also analyze how often elements occur in queries and its impact to parallelization and fault-tolerance of the query system.

%% file: intro.tex

\section{Introduction}

In the \gt 
field,
introduced by 
\cite{dorfman1943detection}, the goal is to identify, by asking queries, all elements of an unknown set $K$. All we initially know about set $K$ is that $|K| \leq k$, for some known parameter $k\le n$, and that it is a subset of some much larger set $N$ with $|N| = n$. The answer to a query $Q$ depends on the intersection between $K$ and $Q$ and equals to $\gfeed(K \cap Q)$, where $\gfeed$ is some 
\dpj{result function (also called feedback function in this paper).}
The sequence of queries is a correct solution to \gt if and only if for any  two different sets $K_1, K_2$ (satisfying some cardinality restriction), the sequence of answers for $K_1$ and $K_2$ is different. Note that this allows to uniquely identify the hidden set $K$ based on 
\dpj{the results of the queries}, though in some cases such decoding could be a hard computational problem. The objective is: 
for a given deterministic feedback function $\gfeed(\cdot)$,
to find a fixed sequence of queries that will identify any set $K$ and the length of this sequence, called the query complexity, will be shortest possible. In particular, we are interested in solutions that have query complexity logarithmic in $n$ and polynomial in $k$. 

The most popular classical 
variant, 
present in the literature,
considers function 
$\gfeed(\cdot)$
that
simply answers whether the intersection between $K$ and $Q$ is empty or not, c.f.,~\cite{duhwang}; 
it is also known under the name of {\em beeping}. 
Another popular 
\dpj{result} 
function
returns the intersection size; this variant has also been studied under the name of \emph{coin weighting}~\cite{Bshouty09, de2013searching} \dpj{and Quantitative Group Testing~\cite{Gebhard19, Feige20}}.
Those variants were applied in many domains, including pattern matching~\cite{clifford2010pattern, Indyk97}, compressed sensing~\cite{cormode2006combinatorial}, streaming algorithms~\cite{cormode2005s, cormode2008finding}, reconstructing graphs~\cite{choi2010optimal,GrebinskiK00}, \dpj{identifying genetic carriers~\cite{cao2014quantitative}}, 
resolving conflicts on multiple-access channels~\cite{capetanakis1979generalized,capetanakis1979tree,gallager1985perspective,greenberg1987estimating,greenberg1985lower,KomlosG85,massey1981collision}.

In this paper we 
study the problem of \gt under a more general \emph{capped quantitative result}, 
\dpj{where the result (feedback) is} 
the size of the intersection up to some parameter $\alpha$ and 
$\alpha$ for larger intersections. It subsumes and generalizes the two previously described classical 
\dpj{result functions}: the smallest possible case of $\alpha = 1$ corresponds to the classical empty/non-empty 
feedback (beeping), 
while the case $\alpha = k$ captures the (full) quantitative 
feedback.
For this generalized 
\dpj{result function} 
we study the influence of the parameter $\alpha$ on the optimal query complexity of \gt\ 
-- thus giving a formal explanation why different settings considered in the literature differ in terms of the query complexity, i.e., the optimal length of query sequence allowing to decode a hidden set~from~feedbacks. 

Our focus is on {\em non-adaptive} solutions, in which queries must be fixed and allow to discover any hidden set based on capped quantitative feedback. Such solutions could be seen as codes: each element corresponds to a binary codeword with $0$ or $1$ on $i$-th position indicating whether the element belongs to $i$-th query. 
Our main emphasis in on {\em minimization of code length} and
\emph{polynomial-time construction of queries/codewords allowing polynomial-time decoding of the hidden set}.
All existing polynomial construction/reconstruction algorithms for beeping or quantitative feedbacks produce codes of length $\Theta(\min\{n,k^2\log n\})$, c.f., \cite{PR11}.
General constructions developed in this work, when instantiated for a specific feedback parameter $\alpha\in [k]=\{1,\ldots,k\}$, shrink the gap
for efficiently constructed query systems even exponentially, and together with the lower bound -- explain why sometimes a much smaller feedback is sufficient
for decoding sets with similar efficiency.

\subsection{Our results}
\label{results}

\begin{table*}
	\centering
	\begin{tabular}{llll}
		\toprule
		$\alpha$ & Constructive Upper Bound & Existential Upper Bound & Lower bound \\\midrule
		$1$ & $O(k^2 \log n)$~\cite{PR11} & $O\left(k^2 \log\frac{n}{k}\right)$~\cite{BonisGV03} & $\Omega\left(k^2 \frac{\log n}{\log k} \right)$~\cite{ClementiMS01} \\\midrule
		$k$ & \begin{tabular}{@{}c@{}}$O(k^2 \log n)$~\cite{PR11} \\ {\bf $\widetilde{O}(k)$ Thm~\ref{thm:constructive-upper}}\end{tabular}  & $O\left(k \frac{\log n}{\log k} \right)$\cite{GrebinskiK00} & $\Omega\left(k \frac{\log n}{\log k} \right)$ (folklore)  \\\midrule
		$*$  &  \begin{tabular}{@{}c@{}} $\widetilde{O}\left(\min\left\{\left(\frac{k}{\alpha}\right)^2,\frac{n}{\alpha} \right\} +k \right)$  \\ {\bf Thm~\ref{thm:constructive-upper}} \end{tabular}
		& \begin{tabular}{@{}c@{}} $O\left( \min\left\{n, \left(\frac{n}{\alpha} + k \right)\log n\right\}\right)$ \\ 
		{\bf if } $k > \sqrt{n \alpha}$,  {\bf Thm~\ref{thm:existential}} \end{tabular} & 
		\begin{tabular}{@{}c@{}} $\Omega \left(\min\left\{\left(\frac{k}{\alpha}\right)^2,\frac{n}{\alpha}\right\} + k\frac{\log n}{\log k} \right)$ \\ 
		{\bf Thm~\ref{thm:lower}}\end{tabular} \\		\bottomrule
	\end{tabular}
	\caption{\label{tab1}  Bounds on query complexity (codeword length) of solutions to non-adaptive \gt with $\cF_{\alpha}$ feedback. By constructive upper bound we mean constructive in time $poly(n)$. Symbol $*$ stands for any valid value of the parameter, notation $\widetilde{O}$ disregards polylogarithmic factors. Our existential upper bound in Theorem~\ref{thm:existential} only covers some range of parameters (it assumes $k > \sqrt{n\alpha}$).}
\end{table*}
We show the first efficient explicit polynomial-time construction of Group Testing query sequence, with an associated polynomial-time decoding algorithm of the hidden set, where the number of queries is only {\em polylogarithmically} far from the absolute lower bound. Previous best polynomial time construction and decoding algorithms, based on superimposed codes, used super-linearly larger query complexity (i.e., with super-quadratic length) than the lower bound. Thus, we shrink the length overhead exponentially, obtaining almost-optimal number of queries by using efficient polynomial-time coding and decoding algorithms.
We also generalize the 
\dpj{result function}
to \emph{capped 
\dpj{result}
 $F_\alpha$}, that returns the size of the intersection only if it is not bigger than some value $\alpha$. We extend our polynomial time algorithms to work for any value of $\alpha$. We also prove a lower bound showing that no other 
\dpj{result function} 
capped at $\alpha$, no matter how complex, could allow substantially less queries (up to a polylogarithmic factor).
One of the consequences of our results is that having quantitative 
\dpj{result function}
capped at $\alpha=\sqrt{k}$ we obtain similar number of queries as with the full 
\dpj{result}
(i.e., returning the size of the whole intersection, up to $k$), which is not a drawback of our method but, as indicated by our lower bound, the inherited property of Group Testing.
\dpj{
We show two applications of our results in streaming. Our first application is an algorithm that processes a stream of insertions and deletions of elements and reconstructs exactly the (multi) set provided that the total number of elements of the set does not exceed $k$. Our second application is an algorithm for maintaining and reconstructing a graph with dynamically added or removed edges.}

We generalize the classical beeping and quantitative feedbacks by defining an {\em $\alpha$-capped quantitative feedback} function, for any $\alpha\in [k]$: 
\[
\cF_{\alpha}(Q \cap K) = \min\{|Q \cap K|,\alpha\}
\ ,
\] 
and study the query complexity of non-adaptive \gt under this feedback, where the query complexity is the number of used queries
 or alternatively -- the length of codewords. We focus on polynomial-time constructing/decoding algorithms.
\paragraph{Main result -- Polynomial-time construction/decoding algorithm using almost optimal number of queries.} 
Here almost-optimality means that the length of the constructed query sequence is only polylogarithmically longer than the shortest possible sequence. The previous best polynomial-time solution used $\Theta(\min\{n,k^2\log n\})$ queries for all $\alpha\ge 1$
\cite{PR11}, 
and we shrink it by factor
$\Theta(\min\{\alpha^2,k\}\,\mbox{polylog}^{-1}\, n)$. 
To achieve this goal, we define and build new types of selectors, called (Strong) Selectors under Interference.
We also generalize the concept of Round-Robin query systems, where each query is a singleton,
to $\alpha$-Round-Robin query systems, 
containing sets of size at most $\alpha$. Such sequences are shorter than the simple Round-Robin, \ie have length $O((n/\alpha)\polylog n)$,
and, unlike a simple Round-Robin singletons' structure, are challenging to construct 
in a way to allow correct decoding based on $\alpha$-capped feedback.
\begin{theorem}
\label{thm:constructive-upper}
There is an explicit polynomial-time algorithm constructing non-adaptive queries $Q_1,\ldots,Q_m$, for $m=O\left(\min\left\{\left(\frac{k}{\alpha}\right)^2 \log^3 n, \frac{n}{\alpha} \polylog n\right\} + k \polylog n\right)$, that solve \gt under feedback 
$\cF_{\alpha}$ with polynomial-time decoding.
Moreover, every element occurs in $\cO(\frac{k}{\alpha}\log^2 n +\polylog n)$ queries, and the decoding time is
$\cO(m+\frac{k^2}{\alpha}\log^2 n +k\polylog n)$.
\end{theorem}

In Section~\ref{sec:applications},
we describe several non-straightforward applications of this result to dynamic graph maintenance and Group Testing on multi-sets; other potential applications include more efficient algorithms for finding hot elements in online streaming~\cite{cormode2006combinatorial} or for wireless communication~\cite{greenberg1987estimating}. %


\paragraph{Lower bound.}  
The almost-optimality of our algorithms from Theorem~\ref{thm:constructive-upper} is justified by proving an absolute lower bound on the length of sequences 
allowing to decode
a hidden set based on feedback $\cF_\alpha$ to the queries.
Here by ``absolute'' we mean that it holds for all query systems that allow for decoding of the hidden sets based on feedback $\cF_\alpha$,
not restricted to polynomially constructed queries with polynomial decoding algorithm. 
Even more, some components of the lower bound are general: they hold for {\em any $\alpha$-capped} feedback function, which will be formally defined later in Section~\ref{sec:model}.
%
The lower bound has three components: $(k/\alpha)^2$, $n/\alpha$, and $k\frac{\log\frac{n}{k}}{\log \alpha}$. For different ranges of $k$, different components determine the value of the lower bound. 
Note that these components match the corresponding components in our constructive upper bound (Main result in Theorem~\ref{thm:constructive-upper}), 
up to polylogarithmic factor.
%
\begin{theorem}
\label{thm:lower}
Any non-adaptive algorithm solving \gt needs:
\begin{itemize}
\item 
$\Omega\left(\min\left\{\left(\frac{k}{\alpha}\right)^2, \frac{n}{\alpha} \right\} + k\frac{\log\frac{n}{k}}{\log \alpha}\right)$ queries under feedback $\cF_{\alpha}$.
\item $\Omega\left(\min\left\{\left(\frac{k}{\alpha}\right)^2, \frac{n}{\alpha} \right\}\right)$ queries under any feedback 
capped~at~$\alpha$.
\end{itemize}
\end{theorem}
Component $k\frac{\log\frac{n}{k}}{\log \alpha}$ in the lower bound follows from a standard information-theoretic argument. To show the remaining parts we use the following idea. Any algorithm working for any set $K$ must ensure that each element $x \in K$ belongs to at least one query $Q$ with small intersection, $|K \cap Q| \leq \alpha$. Otherwise, the element $x$ could be \emph{jammed} by other elements and algorithm would not ``notice'' if it was removed from $K$. With this observation, we obtain that either the algorithm decides to place most of the elements in small queries (\ie of size at most $\alpha$) or each element has to belong to many queries. In the first case we easily obtain 
the necessity of $\Omega(n/\alpha)$ queries. In the second case, we pick set $K_1$ with $k/2$ elements not belonging to any small query and observe that each of the queries to which these elements belong, can be jammed 
by careful selection of $k/2$ elements in $K\setminus K_1$.
This means that for each element $v\in K_1$ we must have at least $\Theta(k/\alpha)$ queries 
containing $v$ and
intersecting $K_1$ on at most $\alpha$ elements, as otherwise element $v$ could be jammed by the remaining $k/2$ elements (in $K\setminus K_1$) by simply choosing $\alpha$ elements from each query to which $v$ belongs. Using this observation we can apply the following counting argument: for each of $k/2$ elements in $K_1$ we have selected $\Theta(k/\alpha)$ queries and each of these queries can be selected at most $\alpha$ times (as otherwise it would be jammed already by set $K_1$). It implies
$\Omega(k^2 / \alpha^2)$ lower bound in this case.


\paragraph{Existential result.}  

Our second upper bound shows the existence of a sequence of queries that provides unique feedback for any set of at most $k$ elements. It is a non-constructive version of $\alpha$-Round-Robin (see our Main result in Theorem~\ref{thm:constructive-upper} for $k \in [\sqrt{n\alpha}, n/2]$), but it has a smaller polylogarithmic factor. 
Its proof is given in Section~\ref{s:upper}. 


\begin{theorem}
\label{thm:existential}
There exists a sequence of queries of length $O\left(\min\{n, \left(k + n/\alpha \right)\log n\}\right)$ solving \gt under feedback $\cF_\alpha$.
\end{theorem}
This upper bound is shown using the probabilistic method by derandomizing three claims that the query sequence has to satisfy simultaneously. The first claim is that each query is small (\ie has at most $\alpha$ elements) and the other two claims ensure that some intersection between queries and a hidden set $K$ will be of size $1$ for two different regimes of parameter~$k$.
\paragraph{Document structure.}
We start from discussing related work on various variants of \gt
in Section~\ref{sec:related-future}. In Section~\ref{sec:model} we formally define the \gt problem and the generalized $\alpha$-capped feedback model. In Section~\ref{sec:poly} we show our polynomial-time construction of an algorithm solving \gt and prove that decoding can be done also in polynomial time. In Section~\ref{sec:applications} we show applications of our constructions to multisets and graph reconstruction. We prove the properties of our new selector tools in section~\ref{sec:tools}. In Section~\ref{s:lower} we show the lower bound, while in 
Section~\ref{s:upper} 
we show the proof of the existential upper bound. 
Discussion of results from perspective of future directions is given in Section~\ref{sec:future}.

%% file: examples.tex
%
\section{Previous and related work}
\label{sec:related-future}

 In the standard feedback model, considered in most of the \gt literature~\cite{duhwang}, the feedback tells whether the intersection between query $Q$ and set $K$ is empty or not  (sometimes it is also called a {\em beeping model}). It it is a special case $\cF_{1}$ of our feedback function. In this feedback model, \gt is known to be solvable using $O(k^2\log(n/k))$~\cite{BonisGV03} queries and 
 an explicit polynomial-time construction of length 
 $O(k^2 \log n)$~\cite{PR11}
 exists. Best known lower bound (for $k < \sqrt{n}$) is $\Omega(k^2 \log n/ \log k )$~\cite{ClementiMS01}. 


The 
setting
considered in this paper is also a generalization of an existing problem of coin weighting. In the coin weighting problem, we have a set of $n$ coins of two distinct weights $w_0$ (true coin) and $w_1$ (counterfeit coin), out of which up to $k$ are counterfeit ones. We are allowed to weigh any subset of coins on a spring scale, hence we can deduce the number of counterfeit coins in each weighting. The task is to identify all the counterfeit coins. Such a feedback is a special case $\cF_{k}$ of our feedback function. The problem is solvable with $O(k\log (n/k)/ \log k)$~\cite{GrebinskiK00} non-adaptive (i.e., fixed in advance) queries and matching a standard information-theoretic lower bound of $\Omega(k\log (n/k)/ \log k)$, as well as its stronger version proved for randomized strategies~\cite{de2013searching}. In~\cite{Bshouty09} the author considers the problem of explicit polynomial-time construction of $O(k\log (n/k)/ \log k)$ queries that allows for polynomial time identification of the counterfeit coins. However, the algorithm presented in~\cite{Bshouty09} is adaptive, which means that the subsequent queries can depend on the feedback from the previous ones. The only existing, constructive, non-adaptive solution would be using the explicit construction of the superimposed codes~\cite{kautz1964nonrandom} but the resulting query complexity would be $O(k^2 \polylog n)$. Thus, the solution presented in our paper is the first explicit polynomial time algorithm constructing non-adaptive queries allowing for fast decoding of set $K$, with $O(k\polylog n)$ 
fixed queries.
 
 \dpj{Our algorithms have direct application in stream processing and set reconstruction. There is a number of existing algorithms retrieving various informations from streams, such as: extracting the most frequent elements~\cite{cormode2003finding, cormode2005s, cormode2008finding,yu2004false}, quantile tracking~\cite{cormode2005holistic, gilbert2002summarize, greenwald2001space}, or approximate histogram maintenance and reconstruction~\cite{gibbons2002fast, gilbert2002fast}. Some of these existing algorithms use Group Testing (e.g.,~\cite{cormode2005s}) but in a randomized variant. This leads to a small probability of error that might become significant in very large streams. }
 
\dpj{
An important line of work on non-adaptive {\em randomized} solutions to Quantitative Group Testing~\cite{Gebhard19, Coja-OghlanGHL20, Feige20, bay2020optimal} resulted in a number of algorithms nearing the lower bound of $2k \frac{\ln(n/k)}{\ln k}$~\cite{djackov1975search}. However, these results always assume some restriction on $k$ (typically $k \sim n^{\theta}$ for some $0<\theta<1$), 
and similarly as above, they may result in significant deviation from the actual set if $n$ is large.
}

The bounds obtained in this paper match (up to polylogarithmic factors) the best existing results for the extreme cases of $\alpha =1$ and $\alpha = k$. The paper also bounds how the query complexity depends on the value of $\alpha$ between these extremes. Interestingly we show, that the shortest-possible query complexity of $k \polylog n$ is already possible for $\alpha = \sqrt{k}$ and increasing $\alpha$ from $\sqrt{k}$ to $k$ does not result in further decrease of the query complexity. 
 

The problem of \gt has also been considered in various different feedback models. For instance,~\cite{Censor-HillelHL15} shows that $O(k \log \frac{n}{k})$ queries are sufficient for a feedback that only returns whether the size of the intersection $|Q\cap K|$ is odd or even. Other interesting feedback function 
is a
\emph{Threshold Group Testing}~\cite{DAMASCHKE}, where the feedback model includes a set of thresholds and the feedback function returns whether or not the size of the intersection is larger or smaller than each threshold. In~\cite{MarcoJK19} the authors show that it is possible to define an interval of $\sqrt{k \log k}$ thresholds resulting in an algorithm with $O(k \log(n / k) / \log k)$ queries.
Note that both those feedbacks are ``inefficient'' in view of our setting of $\alpha$-capped feedbacks, because their feedback functions are not capped at any $\alpha<k$, but they achieve similar query complexity as our capped $\cF_{\sqrt{k}}$ feedback.

Our construction are using known combinatorial tools such as superimposed codes and dispersers. They were used before in Group Testing~\cite{IndykNR10}. However, either it led to a super-quadratic (in $k$) number of queries~\cite{kautz1964nonrandom, cheraghchi2019simple} or decoded only a fraction of elements of the hidden set~\cite{Indyk02}. In solutions, where query complexity depends on the number of identified elements, decoding of all the elements requires over $k^2$ queries~\cite{BonisGV03, ChlebusK05}. In~\cite{IndykNR10}, the authors present the first Group Testing solution with $poly(k, \log n)$ decoding time, but super-quadratic query complexity. It is worth noting that our generalized solution achieves almost-linear number of queries (for certain values of $\alpha$) and our decoding algorithm identifies all the elements in time polynomial in $k$ and logarithmic in $n$. Our use of the known tools is different then in previous approaches: we define new properties (SuI, SSuI), which we prove to be satisfied by some combinations of those tools, and lead to efficient solutions in both query complexity and construction/decoding time.

%% file: model.tex
\section{The model and the problem}
\label{sec:model}
We assume that the universe of all elements $N$, with $|N| = n$, is enumerated with integers $1,2,\dots,n$. Throughout the paper, we will associate an element with its identifier. Let $K$, with $|K| \leq k$, denote a hidden set 
chosen arbitrarily by an adversary. 
Let $\mathcal{Q} = \langle Q_1,\ldots,Q_m \rangle$ be a non-adaptive algorithm, represented by a sequence of queries fixed prior to an execution.

Consider 
feedback function $\cF_\alpha$ that returns the size of an intersection if it is at most $\alpha$ and $\alpha$ for larger intersections (\ie $\cF_\alpha(Q\cap K) = \min\{|Q\cap K|,\alpha\}$). Parameter $\alpha$ in feedback $\cF_\alpha$ is called a {\em feedback cap}. A general class of feedback functions (used in our lower bound) with feedback cap $\alpha$ includes all deterministic functions that take subsets of $[N]$ as input and for sets with more than $\alpha$ output some, arbitrary fixed value.


We will say that $\mathcal{Q}$ solves \gt, if the feedback vector allows for unique identification of set $K$. The feedback vector is defined as:
\[
\langle \cF_\alpha(Q_1\cap K), \cF_\alpha(Q_2\cap K), \dots, \cF_\alpha(Q_t\cap K)  \rangle
\] 
Thus, in order to solve \gt, the feedback vectors for any two sets $K_1$ and $K_2$ have to be different.
We will say that $\mathcal{Q}$ solves \gt with polynomial-time reconstruction if there exists a polynomial-time algorithm that, given the feedback vector outputs all the identifiers of the elements from $K$. Finally we will say that $\mathcal{Q}$ is constructible in polynomial time if there exists a polynomial-time algorithm, that given parameters $n,k,\alpha$ outputs an appropriate sequence of queries.

We assume that both coupled algorithms, construction and decoding, know $n,k,\cF_\alpha$. W.l.o.g., in order to avoid rounding in the presentation, we assume that $n$ and other crucial parameters are powers of $2$.

Alternatively, we can reformulate the problem of non-adaptive \gt under $\cF_\alpha$ feedback into the language of codes. A query sequence translates to code as follows: each element $v \in [N]$ corresponds to a binary codeword with $i$-th position being $1$ or $0$ depending on whether $v$ belongs to the $i$-th query or not. Then the hidden set $K$ is a subset of at most $k$ codewords, for which we calculate the feedback vector by taking the function $\min\{\cdot,\alpha\}$ from elementwise sum of all the codewords corresponding to set $K$. I.e., the feedback is computed for each position $i$, and the whole feedback vector is an input to the decoding algorithm. The objective is to decode the elements of~$K$ from the feedback~vector.

%% file: constructive.tex
%
\section{Polynomial-time constructions and decoding}
\label{sec:poly}
\subsection{Combinatorial tools}

In this section we present combinatorial tools used in our constructions. We introduce two new tools (Selectors-under-Interference and Strong-Selectors-under-Interference) and use one (Balanced IDs) that has previously been used in similar contexts.

\paragraph{Selector under Interference (SuI).}
For given sets $K_1,K_2\subseteq N$ and an element $v\in K_1$, we say that $S\subseteq N$ {\em selects $v$ from $K_1$ under $\alpha$-interference from $K_2$} if $S\cap K_1 = \{v\}$ and $|S \cap K_2|<\alpha$.
Intuitively, $v$ is a unique representative of $K_1$ in $S$ and the number of representatives of $K_2$ in $S$ is smaller than $\alpha$.
An $(n,\ell,\epsilon,\kappa,\alpha)$-Selector-under-Interference, $(n,\ell,\epsilon,\kappa,\alpha)$-SuI for short, is a sequence of queries $\mathcal{S}=(S_1,\ldots,S_x)$ satisfying:
for every set $K_1\subseteq N$ of at most $\ell$ elements and set $K_2\subseteq N$ of at most $\kappa$ elements,
there are at most $\epsilon \ell$ elements $v\in K_1$ that {\em are not} selected from $K_1$ under $\alpha$-interference from $K_2$
by any query $S_i\in \mathcal{S}$, i.e., 
set $\{v\in K_1 \ : \ \forall_{i\le x} \ S_i\cap K_1 \ne \{v\} \mbox{ or } |S_i\cap K_2| \ge \alpha\}$ has less than $\epsilon\ell$ elements.

In Section~\ref{ssec:SuI} we will describe two polynomial-time constructions of SuI and prove the following results.

\begin{theorem}
\label{thm:SuI}
There is an explicit polynomial-time construction of
an $(n,\ell,\epsilon,\kappa,\alpha)$-SuI, 
for any $\ell$, any $\alpha\le k$ such that $\alpha\ell > c_2 \kappa$ for a sufficiently large constant $c_2$, and any constant $\epsilon\in (0,1/2)$, of size 
$\cO(\min\left\{n,\ell \polylog n\right\})$.
Moreover, every element occurs in $\cO(\polylog n)$ queries.
\end{theorem}

\begin{theorem}
\label{thm:SuI-RR}
There is an explicit polynomial-time construction of
an $(n,\ell,\epsilon,\kappa,\alpha)$-SuI, 
for any $\ell$, any $\alpha\le \ell$ and $\ell\le c_2\kappa/\alpha$ 
for a sufficiently large constant $c_2$, 
and any constant $\epsilon\in (0,1/2)$, of size 
$O\left(\min\left\{n,(\kappa/\alpha) \polylog n\right\}\right.$ 
$\left. +\frac{n}{\alpha}\polylog n\right)$.
Moreover, every element occurs in $\cO(\polylog n)$~queries.
\end{theorem}


\paragraph{Strong Selector under Interference (SSuI).}
An $(n,\ell,\kappa,\alpha)$-Strong-Selector-under-Interference, $(n,\ell,\kappa,\alpha)$-SSuI for short, is a sequence of queries $\mathcal{T}=(T_1,\ldots, T_x)$ satisfying:
for every set $K_1\subseteq N$ of at most $\ell$ elements and set $K_2\subseteq N$ of at most $\kappa$ elements,
{\em every} element $v\in K_1$ is selected from $K_1$ under $\alpha$-interference from $K_2$
by some query $T_i\in \mathcal{T}$, i.e., 
set $\{v\in K_1 \ : \ \forall_{i\le x} \ T_i\cap K_1 \ne \{v\} \mbox{ or } |T_i\cap K_2| \ge \alpha\}$ is empty.
An $(n,\ell,\kappa,\alpha)$-Strong-Selector-under-Interference could be also viewed as $(n,\ell,0,\kappa,\alpha)$-Selector-under-Interference.

In Section~\ref{ssec:SSuI} we will describe a polynomial-time construction of SSuI, which essentially is a Kautz and Singleton~\cite{kautz1964nonrandom} construction for adjusted parameters, and prove that it satisfies the additional SSuI property.

\begin{theorem}
\label{thm:SSuI}
There is an explicit polynomial-time construction of
an $(n,\ell,\kappa,\alpha)$-SSuI of length $O(\ell^2 \log_\ell^2 n)$, 
provided $\ell\ge c_2 \kappa/\alpha$ for some constant $c_2$ and for a sufficiently large constant $c>0$. 
Moreover, every element occurs in $\cO(\ell\log_\ell n)$ queries.
\end{theorem}


\paragraph{Balanced IDs.} 
Each element $i$ in $[n]$ has a unique ID represented by $2\log_2 n$ bits, in which the number of 1's is the same as the number of 0's; e.g., take a binary representations of elements $i$ and $n-i$, each in $\log_2 n$ bits, and concatenate them. Balanced IDs have previously been used in algorithms for decoding elements 	in Group Testing (see e.g.,~\cite{lee2019saffron}).

\subsection{Construction of queries, decoding and analysis}


\paragraph{Algorithm constructing queries.}

Let us take 
$(n,\ell,1/2,k,\alpha-1)$-SuI $\cS^{(\ell)}$, 
for $\ell$ being a power of $2$ ranging down from $k$ to 
$c_2 k /(\alpha-1)$
(w.l.o.g. we could also assume that 
$c_2 k /(\alpha-1)$
is a power of $2$). 
Next, for each set $S$ in these selectors we add the following family $\cR(S)=\{R_i(S)\}_{i=1}^{2 \log_2 n}$ 
of sets
$R_i(S)=\{v\in S \ : \ \lfloor v / 2^{i-1}\rfloor = 1 \mod 2\}$. Intuitively $R_i(S)$ is the set of elements from $S$ that have $1$ on $i$-th least significant bit of Balanced ID. 
Let us call the obtained enhanced selectors (i.e., with additional families $\cR(S)$,
for every set $S$ in the original selector) $\bar{\cS}^{(\ell)}$. 
Then we concatenate selectors $\bar{\cS}^{(\ell)}$, starting from the largest $\ell=k$, to the smallest value $\ell=c_2 k/(\alpha-1)$.
An 
$(n,c_2 k/(\alpha-1),k,\alpha-1)$-SSuI $\cT$
is concatenated at the end, with the same replacement of bits 1 and 0 in the original matrix as in the above $(n,\ell,1/2,k,\alpha-1)$-SuI's. 
Algorithm~\ref{alg:construction} presents a pseudocode of the construction algorithm.

\begin{algorithm}[h]
$\ell\leftarrow k, \mathcal{Q} \leftarrow \langle \rangle$\;
\While{$\ell> \frac{c_2l}{\alpha -1}$}{
$\mathcal{S} \leftarrow (n,\ell,1/2,k,\alpha-1)$-SuI\;
\ForEach{$S \in \mathcal{S}$}{
	$\mathcal{Q}$.\textbf{append}$(S)$\;	
	\For{$i\leftarrow 1$ \KwTo $2\log_2 n$}{
	\tcc{Add a set of elements from $S$ that have $1$ on $i$-th least significant bit of Balanced ID.}
	$R_i(S) \leftarrow \{v\in S \ : \ \lfloor v / 2^{i-1}\rfloor = 1 \mod 2\}$\;
	$\mathcal{Q}$.\textbf{append}$\left(R_i(S)\right)$\;
}
}
$\ell \leftarrow \ell/2$\;
}
$\cT \leftarrow (n,c_2 k/(\alpha-1),k,\alpha-1)$-SSuI\;
\ForEach{$T \in \cT$}{
$\mathcal{Q}$.\textbf{append}$(T)$\;	
\For{$i\leftarrow 1$ \KwTo $2\log_2 n$}{
	$R_i(T) \leftarrow \{v\in T \ : \ \lfloor v / 2^{i-1}\rfloor = 1 \mod 2\}$\;
	$\mathcal{Q}$.\textbf{append}$\left(R_i(T)\right)$\;
}
}
\KwRet{$\mathcal{Q}$}
\caption{Construction of a sequence of queries solving Group Testing.
\label{alg:construction}}
\end{algorithm}

\paragraph{Decoding algorithm.}

During the decoding algorithm we process, in subsequent iterations, the feedbacks from enhanced selectors $\bar{\cS}^{(\ell)}$ for $\ell = k,k/2,k/4,\dots, c_2k/(\alpha-1)$. We will later prove, by induction, that during processing $\bar{\cS}^{(\ell)}$, $\ell/2$ new elements from $K$ are decoded. To show this, we consider any iteration and let set $K_1$ be the set of the elements that have been decoded in previous iterations while set $K_2 = K \setminus K_1$ be the set of unknown elements. We treat $K_1$ as the interfering set and, by the properties of $\cS^{(\ell)}$, we know that for at least $\ell/2$ elements $v$, there exists a query $S \in \cS^{(\ell)}$, such that $v\in S$, $|K_1 \cap S| < \alpha - 1$, $|K_2 \cap S| = 1$. We observe that since we already know the identifiers of all the elements from the interfering set $K_1$, then using feedbacks from the additional queries $\cR(S)$ (corresponding to balanced IDs) we can exactly decode the identifier of $v$. We do this for all $\ell/2$ elements that are possible to decode in this iteration and we proceed to the next iteration. After considering all Selectors-under-Interference, we have only at most $c_2k/(\alpha-1)$ unknown elements. To complete the decoding we use a Strong-Selector-under-Interference, where the decoding procedure is exactly the same as in the case of SuI (the interfering set is also the set of already decoded elements). The properties of SSuI guarantee that we decode the identifiers of all the remaining elements from $K$. See the pseudocode of decoding Algorithm~\ref{alg:decoding} for details).

\begin{algorithm}[H]
\KwData{Feedback sequence.}
\KwOut{Set $K$}
\caption{Decoding of the elements.
\label{alg:decoding}}
$\ell \leftarrow k, K_{acc}\leftarrow \emptyset$ \tcc*{In set $K_{acc}$ we accumulate the decoded elements.}
\While{$\ell> \frac{c_2l}{\alpha -1}$}{
\tcc{We want to decode $l/2$ elements.}
\For{$i \leftarrow 1$ \KwTo $l/2$}{
\tcc{Look for queries in $\cS^{(\ell)}$, for which we can decode a new element.}
\ForEach{$S \in \cS^{(\ell)}$}{
\If{$\gfeed(S) < \alpha -1$ \KwSty{and} $|S \cap K_{acc}| = \gfeed(S) - 1$}{
$K_{acc} \leftarrow K_{acc} \cup \{DecodeElement(S,K_{acc})\}$
}
}

}
$\ell \leftarrow \ell / 2$\;

}
\tcc{Decode all the remaining elements.}
\tcc{Iteratively find queries in $\cT$, for which we can decode a new element.}
\ForEach{$T \in \cT$}{
\If{$\gfeed(T) < \alpha -1$ \KwSty{and} $|T \cap K_{acc}| = \gfeed(T) - 1$}{
$K_{acc} \leftarrow K_{acc} \cup \{DecodeElement(T,K_{acc})\}$
}
}
\KwRet{$K_{acc}$}

\setcounter{AlgoLine}{0}
\vspace{1mm}
  \hrule
\vspace{1mm}
  \SetKwProg{myproc}{Procedure}{}{}
  \myproc{DecodeElement($Q$,$K_{acc}$)}{
			$v\leftarrow 0$\;
\For{$j \leftarrow 2\log_2 n$ \KwSty{downto} $0$}{
\tcc{Take the feedback from set $R_j(Q)$. Calculate the feedback from set $R_j(Q)$, if hidden set was exactly $K_{acc}$. The difference is the $j$-th least significant bit of Balanced ID of the new element $v$.}
$v \leftarrow 2\cdot v$\;
$v \leftarrow v + \gfeed(R_j(Q)) - |R_j(Q) \cap K_{acc}|$
}
\KwRet{$v$}
}
\end{algorithm}

\begin{lemma}
\label{lem:smallK}
There is an explicit polynomial-time algorithm constructing non-adaptive queries $Q_1,\ldots,Q_m$ and decoding any hidden set $K$
of size at most $k\le n$, 
from the feedback vector in polynomial time, under feedback 
$\cF_{\alpha}$ 
and for 
$m=O((k/\alpha)^2 \log^3 n+ k\polylog n)$ queries.
Moreover, every element occurs in $\cO(\frac{k}{\alpha}\log^2 n +\polylog n)$ queries.
\end{lemma}

\begin{proof}
We start from describing a procedure of revealing elements in any given set $K$ of at most $k$ elements,
together with a formal (inductive) argument of its correctness.
Our first goal is to show that by the beginning of 
$\bar{\cS}^{(\ell)}$,
for $\ell$ stepping down from $k$ to $c_2 k /(\alpha-1)$,
we have not learned about the identity of at most $\ell$ elements from the hidden set $K$.

The proof is by induction -- it clearly holds in the beginning of the computation,
as the set $K$ has at most $\ell=k$ elements.
We prove the inductive step: by the end of $\bar{\cS}^{(\ell)}$, at most $\ell/2$ elements are not learned. We set $K_2$ to be the set of learned elements and $K_1 = K \setminus K_2$. Clearly, $|K_2| \leq k$, and by the inductive assumption $|K_1| \leq \ell\leq k$.
For such $K_1$ and $K_2$, by the definition of SuI, there are at most $\ell/2$ elements from $K_1$ that are not occurring in some round without other such elements or with at least 
$\alpha -1$
of already learned elements from $K_2$.
Consider a previously not learned element $v \in K_1$, for which there exists a {\em good query} in $\bar{\cS}^{(\ell)}$, \ie 
a query 
$S \in \bar{\cS}^{(\ell)}$ 
such that 
$S \cap K_1 = \{v\}$ and $|S \cap K_2| < \alpha-1$.
At this point 
of decoding of set $K$ 
we know the Balanced IDs of all elements from set $K_2$. Hence we can calculate the $2\log_2 n$-bit feedback vector from sets 
$K_2\cap R_1(S), K_2\cap R_2(S),\dots,K_2\cap R_{2\log_2 n}(S)$.
We compare this feedback vector with the output of the enhanced selector, which is the feedback vector for sets $K \cap R_1(S), K \cap R_2(S),\dots,K \cap R_{2\log_2 n}(S)$.
The difference between the latter and the former is exactly the Balanced ID of $v$. 
In case this difference does not form a Balanced ID of any element, \ie it has some value bigger than $1$ or otherwise the number of 1's is different from $\log_2 n$, or in case $\cF_{\alpha}(K\cap S) = \alpha$ (recall that $S$ is also in the constructed selector)
 the feedback from this $\cR(S)$ is ignored. This is done to avoid misinterpreting the feedback
and false discovery of an element which is not in $K$. Indeed, 
first note that the fact $|K\cap S|\ge \alpha$ will automatically discard the part of the feedback from $K \cap R_1(S), K \cap R_2(S),\dots,K \cap R_{2\log_2 n}(S)$,
as it indicates that the intersection is too large to provide correct decoding of an element. Second, assuming $|K\cap S|< \alpha$, if there are no elements in $K_1\cap S$ then the difference between feedbacks gives vector of zeros, and if there will be at least two elements in $K_1\cap S$, the difference between feedbacks will contain a value of at least $2$ or all 1's, as it will be a 
bitwise sum of at least two Balanced IDs of $\log_2 n$ ones each.
By the definition of SuI we can find $l/2$ such elements $v$. This shows that during decoding of enhanced $\bar{\cS}^{(\ell)}$ we learn the identities of $\ell/2$ new elements. 
This completes the inductive proof. 
Note here that the inductive step, being one of $O(\log n)$ steps, defines a polynomial time algorithm decoding some
elements one-by-one -- indeed, it computes two feedbacks of polynomial number of queries, computes the difference and deducts based on the
structure of subsequent blocks of $O(\log n)$~size.

The above analysis implies, that before applying 
$(n,c_2 k/(\alpha-1),k,\alpha-1)$-SSuI 
we have not discovered at most 
$c_2 k/(\alpha-1)$
elements.
Thus, by definition, 
$(n,c_2 k/(\alpha-1),k,\alpha-1)$-SSuI 
combined with Balanced IDs 
reveals all the remaining elements in the same way as the SuI's above -- the only difference in the argument is that instead
of leaving at most $\ell/2$ undiscovered elements in the $\ell$-th inductive step, due to the nature of SuI's,
the SSuI guarantees that {\em every} undiscovered element will occur in a good query. The same argument as for SuI's
proves that the decoding algorithm defined this way works in polynomial time.

By Theorem~\ref{thm:SuI} below, the length of $(n,\ell,1/2,k,\alpha)$-SuI is \linebreak $\cO(\min\left\{n,\ell \polylog n\right\})$,
which sums up to $\cO(\min\left\{n,k \polylog n\right\})$, and is multiplied by $\Theta(\log n)$ due to amplification by Balanced IDs.
By Theorem~\ref{thm:SSuI}, the length of $(n,c_2 k/\alpha,k,\alpha)$-SSuI is $O((k/\alpha)^2\log^2 n)$,
and it is also increased by factor $\Theta(\log n)$ due to Balanced IDs.
If we apply the above reasoning with respect to the number of queries containing an element, we get that every element occurs in $\cO(\frac{k}{\alpha}\log^2 n +\polylog n)$ queries.
\end{proof}

\paragraph*{Implementing $\alpha$-Round-Robin for large values of $k/\alpha$}
The question arises from the previous result if one could efficiently construct a shorter sequence of queries 
if $(k/\alpha)^2 > n/\alpha$? 
In the case of full feedback (\ie $\alpha=k$) the common way to deal with large values of $k$ is via Round-Robin, which
means that queries are singletons and consequently, the length of such query sequence is $n$. 
This also works for an arbitrary value of $\alpha\le k$, however
the lower bound in Theorem~\ref{thm:lower} and the existential upper bound in Theorem~\ref{thm:existential} 
suggest that in such case there could exist
a shorter query system of length $O((k+(n/\alpha))\polylog n)$. Indeed, if we modify our construction in such case, we could obtain such a goal.
Namely, we concatenate: 
\begin{itemize}
\item
selectors $\bar{\cS}^{(\ell)}$, for $\ell$ being a power of $2$ starting from the largest $\ell=k$ and finishing at $\ell=2c_2\kappa/\alpha$; followed by
\item
selectors $\bar{\cS}|_{\alpha}^{(\ell)}$, for $\ell$ being a power of $2$ starting from the largest $\ell=c_2\kappa/\alpha$
and finishing with $\ell=1$.
\end{itemize}

Then we enhance them based on Balanced IDs, as in the previous construction. Then, applying Theorem~\ref{thm:SuI} for concatenated $\bar{\cS}^{(\ell)}$ and Theorem~\ref{thm:SuI-RR} for concatenated $\bar{\cS}|_{\alpha}^{(\ell)}$, instead of combination of Theorems~\ref{thm:SuI} and~\ref{thm:SSuI} as it was in the proof of Lemma~\ref{lem:smallK} with respect to $\bar{\cS}^{(\ell)}$, we get the following~result.

\begin{lemma}
\label{lem:largeK}
There is an explicit polynomial-time algorithm constructing non-adaptive queries $Q_1,\ldots,Q_m$ and decoding any hidden set $K$
of size at most $k\le n$, 
from the feedback vector in polynomial time, under feedback 
$\cF_{\alpha}$ 
and for 
$m=O((k+n/\alpha)\polylog n)$ queries.
Moreover, each element occurs in  $\cO(\polylog n)$ queries.
\end{lemma}

\begin{proof}
The proof is analogous to the proof of Lemma~\ref{lem:smallK}, except that we continue proving the invariant until $\ell=2c_2\kappa/\alpha$, using same properties guaranteed by Theorems~\ref{thm:SuI}, and continue the invariant until
$\ell=1$, using SuI's of slightly different length formula from Theorem~\ref{thm:SuI-RR}.
The correctness argument, as well as polynomial-time query construction and decoding of the elements, are the same
as in the invariant proof in Lemma~\ref{lem:smallK}.
Then, by Theorem~\ref{thm:SuI} for concatenated $\bar{\cS}^{(\ell)}$ and by Theorem~\ref{thm:SuI-RR} for concatenated $\bar{\cS}|_{\alpha}^{(\ell)}$, we argue that the total length of the obtained
sequence is $m=O((k+n/\alpha)\polylog n)$. Indeed, the first part results from the telescoping sum for different $\ell$
and the second component is a logarithmic amplification of the original $O((n/\alpha)\polylog n)$ length of SuI's; all is amplified
by $O(\log n)$ due to enhancement of the used SuI's by Balanced IDs.
In all the components, every element belongs to $\cO(\polylog n)$ queries, by Theorems~\ref{thm:SuI} and~\ref{thm:SuI-RR},
in the final sequence it also occurs in  $\cO(\polylog n)$ queries.
\end{proof}

Combining Lemma~\ref{lem:smallK} with Lemma~\ref{lem:largeK} gives Theorem~\ref{thm:constructive-upper}.
Note that in both lemmas the decoding algorithm proceeds query-by-query, each time spending polylogarithmic time on each of them; additionally, for each decoded element, an update of the feedback of next queries needs to be done, which takes time proportional to the number of occurrences of the discovered element in the queries. Thus, it is asymptotically upper bounded by the length of the sequence plus $k$ times the upper bound on the number of occurrences of an element in the queries (polylogarithmic).

%% file: applications.tex

\section{Applications}
\label{sec:applications}

\paragraph{Group testing on multi-sets.}

Assume instead of a hidden set, there is a hidden multi-set $K$, containing at most $k$ elements from $[N]$.
Multi-set means that each element may have several multiplicities. Let $\kappa$ be the sum of multiplicities of elements in $K$,
and we assume it is unknown to the algorithm. 
We could decode all elements in $K$ with their multiplicities using similar approach as in Section~\ref{sec:poly},
with the following modifications. 

First, we need to have a sufficiently large cap $\alpha$ to decode each multiplicity, i.e., $\alpha$ should be not smaller than $\kappa$.

Second, in the construction Algorithm~\ref{alg:construction}, instead of applying SuI's only while $\ell> \frac{c_2l}{\alpha -1}$ (line 2) and then SSuI (line 10), we need to keep applying SuI's while $\ell> 1$ (line 2) and remove line 10.
Analogously in the structure of decoding Algorithm~\ref{alg:decoding} -- updating line 2 and removing the end starting from line 8.
This is because the multiplicities of elements not decoded by SuI's in the While-loop could still be larger than $\sqrt{\kappa}$,
therefore switching to SSuI may not be enough to decode their multipicities (note that $\kappa$ plays in this part a similar role to $k$ in the original algorithm for sets without multiplicities).
The correctness still holds, as each consecutive SuI combined with balanced IDs reveals full multiplicities of a fraction of remaining elements in $K$ (instead of just presence of elements, as in the original proof of Theorem~\ref{thm:constructive-upper}).
The asymptotic query complexity (codeword length) stays the same as the part coming from the While-loop (the sum of SuI lengths multiplied by $2\log n$ coming from balanced IDs), since we just add a negligible tail in the sum of lengths of SuI's considered in Theorem~\ref{thm:constructive-upper}. 

Third, polynomial time is now with respect to $n$ and $\log\kappa$, to deal with multiplicities.

Therefore we get:

\begin{theorem}
\label{thm:multi-set-upper}
There is an explicit polynomial-time (in $n$ and $\log\kappa$) algorithm constructing non-adaptive queries $Q_1,\ldots,Q_m$, for $m=O\left(\min\left\{n, k \polylog n\right\}\right)$, that correctly decode a multi-set $K$ of at most $k$ elements and multiplicity $\kappa$ (where $k$ is known but $\kappa$ could be unknown) under feedback 
$\cF_{\alpha}$ with polynomial-time (in $n$ and $\log\kappa$) decoding, where $\alpha$ is not smaller than the largest multiplicity of an element in $K$.
Moreover, every element occurs in $\cO(\polylog n)$ queries, and the decoding time is
$\cO(m+k\polylog n)$.
\end{theorem}


\dpj{
\paragraph{Maintaining and reconstructing a (multi) set.}
Consider the following problem. We have an incoming very large stream of insertion or deletions of elements from some domain $N$. The objective is to propose a datastructure that processes such operations and at any step (i.e., after processing a certain number of operations) it can answer a request and provide information about the set specified by the operations that have been processes so far. This is a commonly studied setting (see e.g.,~\cite{cormode2005s,IlyasAE04}) and extracting information from such stream of operations has applications to database systems. Our algorithms lead directly to an explicit formulation of a datastructure capable of extracting the whole (multi) set, however only conditioned that (at the moment of the request) the sum of multiplicities of the hidden set does not exceed $k$. The space complexity of the datastructure would equal to the number of queries of the algorithm, which is $O\left(\min\left\{n, k \polylog n\right\}\right)$.

}
\paragraph{Maintaining and reconstructing a graph with dynamically added or removed edges.}

Consider a graph $G$ with a fixed set of nodes and an online stream of operations on $G$, where a single operation could be either adding or removing an edge to/from $G$.
Assume for the ease of presentation that after each operation, the maximum node degree is bounded by some parameter~$k$.~\footnote{%
This assumption could be waved by using \gt codes for different parameters $k$, depending on the actual size of $G$, hence $k$ could play role of an average size of a neighborhood.}
Consider a sequence of queries from Theorem~\ref{thm:constructive-upper} on the set of all possible $\frac{n(n-1)}{2}$ edges.
For each added/removed edge, we increase/decrease (resp.) a counter associated with each query containing this edge.
As each edge occurs in $\cO(\frac{k}{\alpha}\log^2 n +\polylog n)$ queries, and thus this is an upper bound (up to some additional logarithmic factor) on the time of each {\em graph update},
which is {\em polylogarithmic} for $\alpha$ close to $k$.
Whenever one would like to recover the whole graph, a {\em reconstruction} algorithm is applied, which takes 
$\cO(m+\frac{k^2}{\alpha}\log^2 n +k\polylog n)$ steps,
which for $\alpha$ close to $k$ is $\cO(nk\polylog n)$. Note that the latter formula corresponds to (the upper bound on) the number of edges in $G$. To summarize, we implemented graph updates operations in $\polylog n$ time per (edge-)operation, and the whole graph recovery in time proportional to the graph size (number of edges) times $\polylog n$.

\dk{
\paragraph{Private Parallel Information Retrieval (PPIR)}
One of techniques to speed-up Information Retrieval from a large dataset is to employ autonomous agents searching parts of the datasets, c.f.,~\cite{GalanisWJD03}. Our Capped Quantitative Group Testing algorithms could be applied to achieve this goal, additionally providing a level of privacy. Assume that there are $m=O(k\polylog n)$ simple autonomous agents, where $m$ is the number of queries in our Capped QGT system. Each agent $i$ is capable to search only through records captured by the corresponding query set $Q_i$, and only count the number of occurrences of records satisfying the search criteria, but only up to $\sqrt{k}$. If all agents share privately their results with the user, he can decompose the set $K$ of at most $k$ elements satisfying the searching criteria, while each of the agents has knowledge about at most $\sqrt{k}$ of these elements. It follows from the construction of our queries that each of them is of size $O(n/\sqrt{k} \polylog n)$, which is worst-case number of records that a single agent needs to check -- thus equal to parallel time. Note also that agents perform very simple counting operations, thus the PPIR algorithm could be efficient in practice. 
}

%% file: selectors.tex
%
\section{Constructions of combinatorial tools}
\label{sec:tools}
\subsection{Polynomial-time construction of Selectors-under-Interference}
\label{ssec:SuI}

In this section, we show how to construct, in time polynomial in~$n$, an $(n,\ell,\epsilon,\kappa,\alpha)$-SuI $\mathcal{S}$ of size
$\cO(\min\left\{n,\ell \polylog n\right\})$, 
for any integer parameters $\ell,\kappa\le n$, $\alpha\le \kappa$, and any (arbitrarily small) constant $\epsilon\in (0,1/2)$.
Let $\ell^*$ denote~$\ell \ep$.
The construction combines dispersers with strong selectors, see also the pseudocode Algorithm~\ref{alg:SuI}. 
We start from specifying those tools.

\subparagraph{Disperser.}
Consider a bipartite graph $G=(V,W,E)$, where $|V|=n$, 
which is an {\em $(\ell^*,d,\ep)$-disperser with entropy loss $\delta$}, 
i.e., it has left-degree $d$, 
$|W|=\Theta(\ell^* d/\delta)$, 
and satisfies the following dispersion condition:
for each $L\subseteq V$ such that $|L|\ge \ell^*$,
the set $N_G(L)$ of neighbors of $L$ in graph $G$ is of size at least $(1-\ep)|W|$.
%
Note that it is enough for us to take as $\epsilon$ in the dispersion property the same value as
in the constructed $(n,\ell,\epsilon,\kappa,\alpha)$-SuI $\mathcal{S}$.
An explicit construction (i.e., in time polynomial in~$n$) of dispersers was given by
Ta-Shma, Umans and Zuckerman~\cite{TUZ},
for any $n\ge \ell$, and some $\delta=\cO(\log^3 n)$,
where $d=\cO(\text{polylog }n)$. 

\subparagraph{Strong selector.}
Let $\cT=\{T_1,\ldots,T_m\}$ be  an explicit 
$(n,c\delta)$-strong-selector (also called strongly-selective family), 
for a sufficiently large constant $c>0$ that will be fixed later, 
of size \linebreak
$m=\cO(\min\{n,\delta^2\log^2 n\})$,
as constructed by Kautz and Singleton~\cite{kautz1964nonrandom}.

\subparagraph{Construction of $(n,\ell,\epsilon,\kappa,\alpha)$-SuI $\mathcal{S}$.}
We define an $(n,\ell,\epsilon,\kappa,\alpha)$-SuI $\mathcal{S}$
of size
$\min\{n,m|W|\}$, which consists of sets $S_i$, 
for $1\le i\le \min\{n,m|W|\}$.
There are two cases to consider, depending on the relation
between $n$ and $m|W|$.
The case of $n\le m|W|$ is simple: take the singleton 
containing only the $i$-th element of~$V$ as~$S_i$.
Consider a more interesting case when $n> m|W|$.
For $i=am+b\le m|W|$, where $a$ and $b$ are non-negative integers 
satisfying $a+b>0$, let $S_i$ contain all the nodes $v\in V$ 
such that $v$ is a neighbor of the $a$-th node in~$W$ and $v\in T_b$. 

\begin{algorithm}[h]
\KwData{$(\ell\epsilon,d,\ep)$-disperser $G = (V,W,E)$, $V = \{v_1,\dots,v_n\}$, $W = \{w_1,\dots,w_{|W|}\}$, $|W| = \Theta(\ell d/\delta)$, $\delta=\cO(\log^3 n)$, $d=\cO(\text{polylog }n)$, \\ \hspace{7mm} $(n,c\delta)$-strong-selector $\cT =\{T_1,\ldots,T_m\}$}
\KwResult{$(n,\ell,\epsilon,\kappa,\alpha)$-SuI $\mathcal{S}$}
\For{$i\leftarrow 1$ \KwTo $\min\{n,m|W|\}$}{
\If{$n > m|W|$}{
$S_i \leftarrow \{v_i\}$ 
}
\Else{
Find $a,b > 0$, such that $i=am+b\le m|W|$ \;
$S_i \leftarrow T_b \cap N_G(w_a)$ 
}
}
\KwRet{$\langle S_1,S_2,\dots,S_{\min\{n,m|W|\}}\rangle$}
\caption{Construction of Selectors-under-Interference (SuI).
\label{alg:SuI}}
\end{algorithm}

\remove{
Strongly selective families are $(n,k,k)$-selectors.
We show how to use dispersers to decrease the third
parameter~$r$ in $(n,k,r)$-selectors while also gracefully 
decreasing the size of the family of sets.

If $r\le 3k/4$ then we can use the construction of an
$(n,k,3k/4)$-selector given by Indyk~\cite{Ind}.
Assume that $r>3k/4$.
Let $0<\ep<1/2$ be a constant.
}

\remove{
\begin{theorem}
\label{thm:SuI}
The family 
$\cS$ is an $(n,\ell,\epsilon,\kappa,\alpha)$-SuI, 
for any $\ell$, any $\alpha\le k$ such that $\alpha\ell > c_2 \kappa$ for a sufficiently large constant $c_2$, and any constant $\epsilon\in (0,1/2)$, of size 
$\cO(\min\left\{n,\ell \polylog n\right\})$.
Moreover, every element occurs in $\cO(\polylog n)$ queries.
\end{theorem}
}

\begin{proof}[Proof of Theorem~\ref{thm:SuI}]
First we show that 
the constructed $\cS$ is an $(n,\ell,\epsilon,\kappa,\alpha)$-SuI.
The case $n\le m|W|$ is clear, since each element in a set~$K_1$ of
size at most~$\ell$ occurs as a singleton in some set~$S_i$ (here it does not matter what the set $K_2$ is).

Consider the case $n>m|W|$.
Let a set $K_1\subseteq V$ be of size at most~$\ell$ and a set $K_2$ of at most $\kappa$ elements.
Suppose, to the contrary, that there is a set $L\subseteq K_1$ of 
size $\ell^*$ such that none among the elements in~$L$
is $K_1$-selected by $\cS$ under $\alpha$-interference from $K_2$, 
that is, $S_i\cap L\ne \{v\}$ or $|S_i\cap K_2|\ge \alpha$, for any $v\in L$ and $1\le i\le m|W|$.
(Recall that $\ell^* = \ell \ep$.)
%

\noindent
\textsf{Claim:}
Every $w\in N_G(L)$ has more than 
$c\delta$ 
neighbors in~$K_1$
or at least $\alpha$ neighbors in $K_2$. 

The proof is by contradiction.
Suppose, to the contrary, that
there is $w\in N_G(L)$ which has
at most 
$c\delta$ 
neighbors in~$K_1$ and less than $\alpha$ neighbors in $K_2$,
that is, $|N_G(w)\cap K_1|\le c\delta$ and 
$|N_G(w)\cap K_2|<\alpha$. 
By the former property and the fact that $\cT$ is an 
$(n,c\delta)$-strong-selector, 
we get that, for every $v\in N_G(w)\cap K_1$, the equalities
\[
S_{w\cdot m+b}\cap K_1 = 
(T_b\cap N_G(w)) \cap K_1 =
T_b\cap(N_G(w)\cap K_1) =
\{v\} 
\]
hold, for some $1\le b\le m$.
This holds in particular for every $v\in L\cap N_G(w)\cap K_1$.
There is at least one such~$v\in L\cap N_G(w)\cap K_1$ 
because set $L\cap N_G(w)\cap K_1$ is nonempty 
since $w\in N_G(L)$ and $L\subseteq K_1$.
Additionally, recall that $N_G(w)\cap K_2$ is smaller than $\alpha$.
The existence of such~$v$ is in contradiction with the choice of~$L$.
Namely, $L$ contains only elements which are not $K_1$-selected
by sets from $\cS$ under $\alpha$-interference from $K_2$,
but $v\in L\cap N_G(w)\cap K_1$ 
is selected from $K_1$ by some set 
$S_{w\cdot m+b}$ and the interference from $K_2$ on this set is smaller than $\alpha$. 
This makes the proof of Claim complete. $\blacksquare$

Recall that $|L|=\ell^*=\ell \ep$.
By dispersion, the set $N_G(L)$ is of size larger than $(1-\ep)|W|$.
Consider two cases below -- they cover all possible cases because of the above Claim.

\noindent
\textsf{Case 1:} At least half of the nodes $w$ in $N_G(L)$
have more than $c\delta$ neighbors in~$K_1$.

In this case, the total number of edges 
between the nodes in $K_1$ and 
$N_G(L)$ in graph $G$ is larger than 
\[
\frac{1}{2}(1-\ep)|W|\cdot c\delta =
\frac{1}{2}(1-\ep)\Theta(\ell d/\delta) \cdot c\delta > 
\ell d \ ,
\]
for a sufficiently large constant~$c$. 
This is a contradiction, since the total number of edges in graph $G$ 
incident to nodes in~$K_1$ is at most $|K_1|d= \ell d$. 

\noindent
\textsf{Case 2:} More than half of the nodes $w$ in $N_G(L)$
have at least $\alpha$ neighbors in $K_2$.

In this case, the total number of edges 
between the nodes in $K_2$ and 
$N_G(L)$ in graph $G$ is larger than 
\[
\frac{1}{2}(1-\ep)|W|\cdot \alpha =
\frac{1}{2}(1-\ep)\Theta(\ell 
d/\delta) \cdot \alpha >
\kappa d  \ ,
\]
for a sufficiently large constant~$c_2$.
This is a contradiction, since the total number of edges in graph $G$ 
incident to nodes in~$K_2$ is at most $|K_2|d= \kappa d$. 

Thus, it follows from the contradictions in both cases that $\cS$ is an $(n,\ell,\epsilon,\kappa,\alpha)$-SuI.

The size of this selector is
\begin{eqnarray*}
\min\{n,m|W|\}
&=&
\cO\left(\min\left\{n,\delta^2\log^2 n \cdot \ell^*d/\delta\right\}\right) 
\\
&=&
\cO\left(\min\left\{n,\ell^*\delta d\log^2 n \right\}\right) \\
&=&
\cO\left(\min\left\{n,\ell\polylog n\right\}\right) \ ,
\end{eqnarray*}
since $d=\cO(\polylog n)$, $\delta=\cO(\log^3 n)$ and $\ell=\Theta(\ell^*)$.
It follows directly from the construction that each element is in $\cO(d\delta^2\log^2 n)=\cO(\polylog n)$ queries.
\end{proof}

\paragraph*{Sparser SuI for small $\ell$ compared to $\kappa/\alpha$}

What if $\alpha\ell \le c_2 \kappa$ for some constant $c_2>0$?
We could modify the above construction as follows.
Let $\gamma=\frac{\alpha\ell}{\kappa}$.
If $\gamma\le c_2$,  we take 
the $(n,c_2\kappa/\alpha,\epsilon,\kappa,\alpha)$-SuI $\cS$ from Theorem~\ref{thm:SuI} and partition each set $S_i\in \cS$ into the smallest number of sets of size at most $\alpha$ each. Note that the total number of occurrences of elements in sets $S_i\in\cS$ in the above construction is upper bounded by the number of edges in the disperser multiplied by the number of occurrences of elements in the strong selector,
which asymptotically gives $O(nd\cdot \delta\log n)=O(n \polylog n)$.
Therefore, after the above partition of sets $S_i$, the total number of sets in the obtained sequence is 
$O\left(\min\left\{n,(\kappa/\alpha) \polylog n\right\}+\frac{n}{\alpha}\polylog n\right)$. We denote the new sequence obtained from $\cS$ by
$\cS|_{\alpha}$
Note that it is an $(n,\ell,\epsilon,\kappa,\alpha)$-SuI, as if in the original $(n,c_2\kappa/\alpha,\epsilon,\kappa,\alpha)$-SuI $\cS$ an element $v$ was $\alpha$-selected from a  set $K_1$ under interference from $K_2$, where $|K_1|\le \ell \le c_2\kappa/\alpha$, it occurs in some of the new sets being in the partition of the original selecting set, and by monotonicity of selection -- $v$ is also $\alpha$ selected from $K_1$ under interference from $K_2$. 
Note that in the above decoding, the number of queries containing any element remains $\cO(\polylog n)$ as in original SuI from Theorem~\ref{thm:SuI}.
Hence, by taking the construction of family 
$\cS|_{\alpha}$, we proved Theorem~\ref{thm:SuI-RR}.
\remove{
\begin{theorem}
\label{thm:SuI-RR}
The family 
$\cS|_{\alpha}$ is an $(n,\ell,\epsilon,\kappa,\alpha)$-SuI, 
for any $\ell$, any $\alpha\le \ell$ and $\ell\le c_2\kappa/\alpha$ 
for a sufficiently large constant $c_2$, 
and any constant $\epsilon\in (0,1/2)$, of size 
$O\left(\min\left\{n,(\kappa/\alpha) \polylog n\right\}\right.$ 
$\left. +\frac{n}{\alpha}\polylog n\right)$.
Moreover, every element occurs in $\cO(\polylog n)$ queries.
\end{theorem}
}

%% file: strong.tex
%
\subsection{Polynomial-time construction of Strong-Selectors-under-Interference}
\label{ssec:SSuI}

In order to construct an $(n,\ell,\kappa,\alpha)$-SSuI  $\mathcal{S}$,
we 
%
use the following variation of a Reed-Solomon superimposed code, analogous to the construction used in~\cite{kautz1964nonrandom}, however here we prove an additional property of these objects.
\begin{enumerate}
\item
Let 
$d= \left \lceil \log_\ell n \right \rceil$ and $q=c \cdot \ell \cdot d$ 
for some constant $c>0$ such that $q^{d+1} \ge n$ and $q$~is~prime.

\item
Consider all polynomials 
of degree $d$ over field $\fieldF_q$; there are $q^{d+1}$ such polynomials. 
Remove $q^{d+1}-n$ arbitrary polynomials and denote the remaining polynomials by $P_1,P_2,\dots,P_n$.
\label{step:remove}

\item
Create the following matrix $M$ of size 
$q \times n$.
Each column $i$, for $1\le i\le n$, stores values $P_i(x)$ of polynomial $P_i$ for arguments $x = 0, 1, \dots, q-1$;
the arguments correspond to rows of $M$. 
Next, matrix $M^*$ is created from $M$ as follows: each value $y=P_i(x)\in\{0,1,\ldots,q-1\}$ is represented 
and padded in $q$ consecutive rows containing $0$s and $1$s, where $1$ is exactly in $y+1$-st padded row, while in all other padded rows there are $0$s. Notice that each column of $M^*$ has $q^2$ rows ($q$ rows per each argument),
therefore $M^*$ is of size $q^2\times n$.

\item
Set $T_i\subseteq [n]$, for $1\le i\le q^2$, is defined based on row $i$ of matrix $M^*$: it contains all elements $v\in [n]$ such that $M^*[j,v]=1$. (Recall that each such $v$ corresponds to some polynomial.)
For a fixed constant $c > 0$, $\{T_i\}_{i=1}^{q^2}$ forms a family $\mathcal{T}^{(c)}$ of subsets of set $\{1,\ldots,n\}$. 
\end{enumerate}

The above construction could be presented as a simplified pseudocode as follows:
\begin{algorithm}[h]
$d \leftarrow \left \lceil \log_\ell n \right \rceil$\;
$q \leftarrow  c \cdot \ell \cdot d$ for some constant $c>0$ such that $q^{d+1} \ge n$ and $q$~is~prime\;
\tcc{There are $q^{d+1} \ge n$ such polynomials.}
$P_1,P_2,\dots,P_n \leftarrow $ arbitrary $n$ polynomials of degree $d$ over field $\fieldF_q$\;
$\mathcal{T}^{(c)} \leftarrow $ sequence of $q^2$ empty sets $\{T_i\}_{i=1}^{q^2}$\;
\For{$i\leftarrow 1$ \KwTo {$n$}}{
\For{$x\leftarrow 0$ \KwTo {$q - 1$}}{
\tcc{Value of $i$-th polynomial for argument $x$.}
	$value \leftarrow P_i(x)$\;
\tcc{Encode $value$ in unary on positions $x\cdot q + 1, x\cdot q+2,\dots (x+1)\cdot q$}
	$index \leftarrow x \cdot q + value + 1$\;
	\tcc{Add element $i$ to the corresponding set $T$.}
	$T_{index}$\textbf{.add}$(i)$\
}
}
\KwRet{$\mathcal{T}^{(c)}$}
\caption{Construction of Strong-Selectors-under-Interference}
\end{algorithm}

\remove{
\begin{theorem}
\label{thm:SSuI}
$\mathcal{T}^{(c)}$ is an $(n,\ell,\kappa,\alpha)$-SSuI of length $O(\ell^2 \log_\ell^2 n)$ constructed in polynomial time,
provided $\ell\ge c_2 \kappa/\alpha$ for some constant $c_2$ and for a sufficiently large constant $c>0$. 
Moreover, every element occurs in $\cO(\ell\log_\ell n)$ queries.
\end{theorem}
}
\begin{proof}[Proof of Theorem~\ref{thm:SSuI}]
Consider the constructed family $\mathcal{T}^{(c)}$.
Polynomial time of this construction follows directly from the fact that the space of polynomials over field $[q]$ is of polynomial size in $n$
and all the operations on them are polynomial. The length follows from the fact that it is $q^2=O(\ell^2 \log_\ell^2 n)$.

Recall that each element $v\in [n]$ correspond to some polynomial of degree at most $d$ over $\fieldF_q$.
Note that two polynomials $P_i$ and $P_j$ of degree $d$ with $i \neq j$, can have equal values for at most $d$ different arguments. This is because they have equal values for arguments $x$ for which $P_i(x) - P_j(x)=0$. However, $P_i - P_j$ is a polynomial of degree at most $d$, so it can have at most $d$ zeroes. Hence, $P_i(x) = P_j(x)$ for at most $d$ different arguments~$x$.

Take any polynomial $P_i$ and any other at most $\ell-1$ polynomials $P_j$, which altogether form set $K_1$ of at most $\ell$ 
polynomials.
There are at most $(\ell-1) \cdot d$ different arguments where one of the other $\ell-1$ polynomials can be equal to $P_i$. 
Hence, for at least $q-(\ell-1) \cdot d$ different arguments, the values of the polynomial $P_i$ are different than the values of the other 
polynomials in $K_1$. Let us call the set of these arguments $A$.

Consider any set $K_2\subseteq [n]$ of at most $\kappa$ elements (corresponding to polynomials).
Consider arguments from set $A$ for which $P_i$ has the same value as at least $\alpha$ other polynomials in $K_2$.
The number of such arguments is at most 
\[
\frac{\kappa\cdot d}{\alpha} \le
(\ell/c_2) \cdot d
< 
(c-1)\ell \cdot d
< q - (\ell-1)\cdot d
\ ,
\]
which means it is smaller than $|A|$ for sufficiently large constant $c>0$ in the definition of $q=c\ell \cdot d$.
Therefore, there is an argument (in set $A$) such that the value of $P_i$ is different from the values of all other $\ell-1$ 
polynomials in $K_1$ {\em and} less than $\alpha$ polynomials in set $K_2$. 
As this holds for an arbitrary polynomial $P_i$ in an arbitrary set $K_1$ of at most $\ell$ polynomials (in total)
and an arbitrary set $K_2$ of at most $\kappa$ polynomials, $\cT^{(c)}$ is an $(n,\ell,\kappa,\alpha)$-SSuI.
\remove{
Therefore, if we look at rows with $1$ in column $i$ of matrix $M$ (there are $q$ of those rows, one for each argument), at least $q-k \cdot d$ of them have $0$s in chosen $k$ columns. Since there are $q^2$ rows, so a fraction $(q-k \cdot d)/q^2$ of rows have the desired property (i.e., there is value $1$ in column $i$ and value $0$ in the chosen $k$ columns):

$$\dfrac{q-k \cdot d}{q^2} = \dfrac{(c-1) \cdot k \cdot d}{(c \cdot k \cdot d)^2} = \dfrac{c-1}{c^2 \cdot k \cdot d} \triangleq f(c) \ . $$

Let us find the value of $c$ that maximizes the function $f$. To do it, we compute its differential 

$$f'(c) = (\dfrac{c-1}{c^2 \cdot k \cdot d})' = \dfrac{1\cdot(c^2 \cdot k \cdot d)-(c-1)\cdot k \cdot d \cdot 2c}{c^4 \cdot k^2 \cdot d^2} =$$ 
$$= \dfrac{-c^2 \cdot k \cdot d + 2c \cdot k \cdot d}{c^4 \cdot k^2 \cdot d^2} = \dfrac{-c + 2}{c^3 \cdot k \cdot d} \ .$$ 

Thus, $f'(c)=0$ for $c=0$ or $c=2$. The value $c=2$ maximizes $f$, giving 
$f(c)\le f(2) = 1/(4k \cdot d) = 1/(4k \cdot \log_k n)$.

Therefore, we can construct a $(n,k,\epsilon)$-universally-strong selector with $\epsilon = f(2) \cdot k = 1/(4d) = 1/(4 \log_k n)$ of length $4k^2 \cdot \left\lceil\log_k n \right\rceil^2$ (which means that an $f(2) = 1/(4k \cdot \log_k n)$ fraction of the selector's sets have the desired property).}
Finally, it follows directly from the construction that every element occurs in $q=\cO(\ell\log_\ell n)$ queries.
\end{proof}

%% file: size-feedback-lower.tex
%
\section{Lower bound}
\label{s:lower}
\begin{proof}[Proof of Theorem~\ref{thm:lower}]
We will first show the $\min\left\{\frac{n}{\alpha},\frac{k^2}{\alpha^2} \right\}$ component. Assume that a sequence of queries $Q_1,Q_2,\dots,Q_t$ of length $t$ solves \gt. We want to show the lower bound that holds for any feedback function capped at $\alpha$ hence we assume that the feedback function $\mathcal{F}$ returns the whole set (i.e., the identifiers of all the elements). Recall that $\mathcal{F}$ works only for sets with at most $\alpha$ elements. 
We begin by proving the following:
\\
\textsf{Claim A:} For any set $K$, with $|K| \leq k$ and any $x \in K$, there must exist $\tau \in \{1,2,\dots, t\}$, such that $x \in Q_{\tau}$ and $|K \cap Q_{\tau}|\leq \alpha + 1$.

The proof is by contradiction. Assume that such a set $K^*$ and element $x^*$ exist for which there is no such query. Consider feedback vectors for sets $K^*$ and $K^* \setminus \{x^*\}$. For any query that does not contain $x^*$, the feedback is clearly identical. For any query $Q_{\tau}$, such that $x^*\in Q_{\tau}$, we have $|Q_{\tau} \cap K^*|\geq \alpha + 2$ and $|Q_{\tau} \cap (K^*\setminus \{x^*\})|\geq \alpha + 1$ and 
sets $K^*$ and $K^*\setminus \{x^*\}$ are indistinguishable under any feedback capped at $\alpha$ hence the sequence of queries does not solve the problem. This completes the proof of Claim A. \ $\blacksquare$
 
 Take all queries that have at most $\alpha +1$ elements and all elements that belong to such queries.
 We have:
 $
 N_{s} = \bigcup_{\tau \in \{1,2,\dots,t\} \\ |Q_{\tau}| \leq \alpha +1} Q_{\tau} \ .
 $ 
Denote the remaining elements by $N_{l} = N \setminus N_{s}$. We will consider two cases:
\\
\textsf{Case $1$:} $|N_{s}| \geq n/2$

Observe that: 
$t \geq |\{\tau \in \{1,2,\dots,t\} : |Q_{\tau}| \leq \alpha + 1\}| \geq \frac{N_{s}}{\alpha+1} \geq \frac{n}{2(\alpha+1)}$. 

\noindent \textsf{Case $2$:} $|N_{l}| \geq n/2$

 In this case, we take an arbitrary subset $K_1$ of $k/2$ elements from $N_{l}$. For every element $x \in K_1$, we consider a set of queries $Q(x) = \{Q_{\tau} \in \{Q_1,Q_2,\dots,Q_t\}:  x \in Q_{\tau}, |Q_{\tau} \cap K_1| \leq \alpha +1\}$. We first show the following:
 \\
 \textsf{Claim B:} For every $x \in K_1$, we have $|Q(x)| \geq \frac{k}{2(\alpha+2)}$.
 
The proof is by contradiction. Assume that for some $x^* \in K_1$ we have $|Q(x^*)| < \frac{k}{2(\alpha + 2)}$. Then, for every query $Q \in Q(x^*)$, we take $\alpha + 2 - |Q \cap K_1|$ elements from $Q\setminus K_1$. Such elements exist since $|Q| \geq \alpha + 2$. Choose such elements for each query in $Q(x^*)$ and gather them in set $K_2$. Note that since $|Q(x^*)| < \frac{k}{2(\alpha+2)}$, then $|K_2| \leq k/2$. Now observe that set $K_1 \cup K_2$ and element $x^*$ violate Claim A. The obtained contradiction completes the proof of Claim  B. \ $\blacksquare$

 Now observe that each query belongs to at most $\alpha+1$ sets $Q(x)$ for different values of $x\in K_1$. Thus:
 $t \geq \frac{\sum_{x\in K_1} |Q(x)|}{\alpha+1} \geq \frac{k^2}{4(\alpha+1)(\alpha+2)}$. 
 
To complete the proof observe that any algorithm must fall either into Case 1 or Case 2, hence any algorithm needs to use $\Omega\left(\min\left\{\frac{n}{\alpha},\frac{k^2}{\alpha^2} \right\}\right)$ queries.

To see that any algorithm in $\cF_\alpha$ feedback model at least $k\frac{\log\frac{n}{k}}{\log \alpha}$ queries, observe that the feedback vector must be unique for each set $K$ with at most $k$ elements. Hence we need at least ${n\choose k}$ different feedback vectors for different sets. Feedback has at most $\alpha$ values hence we get $\alpha^t \geq {n \choose k}$ and $t \in \Omega(k\frac{\log\frac{n}{k}}{\log \alpha})$. 
\end{proof}

%% file: largeK.tex
\section{Existential upper bound}
\label{s:upper}

\begin{proof}[Proof of Theorem~\ref{thm:existential}]

Assume that $\alpha >2 \log_2 n$, the opposite case will be considered at the end of the proof. We will prove using the probabilistic method that a $\alpha$-Round-Robin sequence of queries of length $t = O((n/\alpha + k)\log n)$ exists. We take $t_1 = \left\lceil\left(\frac{8n}{\alpha}\right)(\ln(ne)+ 4)\right\rceil$, $t_2 = \left\lceil k (\ln(ne)+ 4)\right\rceil$, $t = t_1 + t_2$ and construct a sequence of queries $\mathcal{Q} = \langle Q_1,Q_2,\dots, Q_t\rangle$ as follows. For $i\in[1,t_1]$, Each query $Q_i$ is constructed by including each element from $N$ independently at random with probability $p = \frac{\alpha}{6 n}$. For $i \in [t_1+1,t_1 + t_2]$, each query $Q_i$ is constructed by including each element from $N$ independently at random with probability $p = \min\{\frac{1}{6k}, \frac{\alpha}{6 n}\}$. Denote the first $t_1$ queries by $\mathcal{Q}_1$ and the remaining queries by $\mathcal{Q}_2$.

\noindent\textsf{Claim 1:}  With probability at least $2/3$ each query in $\mathcal{Q}$ has at most $\alpha$ elements.

Take any query $Q \in \mathcal{Q}$ and observe that the size of the query is a sum of Bernoulli trials and $E{|Q|} \leq \alpha / 6$. Using Chernoff bound~\cite{SURV}, since $\alpha > 6E{|Q|}$ we get:
$
\Pr{|Q| \geq \alpha} \leq 2^{-\alpha} \leq \frac{1}{n^2} \ .
$ 
Hence, knowing that $\alpha > 2\log_2 n$ the probability that any query is larger than $\alpha$ is by the union bound at most $t/n^2 <1/3$.

\noindent\textsf{Claim 2:} With probability at least $3/4$, for any set $K$, with $|K| \leq k$, for $k \leq n/\alpha$, some query $Q \in \mathcal{Q}_1$, satisfies $|Q \cap K| = 1$.

Consider any query $Q \in \mathcal{Q}_1$ and set $K$. Let $k^* = |K|$. We know that $k^* \leq n / \alpha$. 
We have:
\[
\Pr{|Q \cap K| = 1} = k^* \cdot \frac{\alpha}{6n} \cdot \left(1-\frac{\alpha}{6n}\right)^{k^*-1} \geq  \frac{k^*\alpha}{6n} \cdot \left(1-\frac{k^*\alpha}{6n}\right) \geq \frac{k^* \alpha}{8n} \ .
\]
Hence if $k^* \in [2^i,2^{i+1}]$, then $\Pr{|Q \cap K| = 1}  \geq \frac{2^i \alpha}{8n}$. We want to union bound the probability that the sequence fails to select some element from set $K$ over all possible sets $K$. We denote event $\mathsf{fail}$ as the event that $\mathcal{Q}_1$ fails to hit any set with a most $k$ elements. The possible number of sets of $K$ with $k^* \in [2^i,2^{i+1}]$ is at most $2^{i+1} {n \choose 2^{i+1}}$. Thus:
\[
\Pr{\mathsf{fail}} \leq \sum_{i=0}^{\log k} 2^{i+1} {n \choose 2^{i+1}} \left(1 - \frac{2^i \alpha}{8n}\right)^{t_1} \leq \sum_{i=0}^{\log k} e^{i+1} \cdot e^{2^i \ln \frac{ne}{2^i}} e^{-t_1 \cdot \frac{2^i \alpha}{8n }} \ .
\]
Knowing that $t_1 > \frac{8n}{\alpha}(\ln(ne)+ 4)$, we have for any $i\geq 0$,
$
t_1 > \frac{8n}{\alpha}\left(\ln\frac{ne}{2^i} + \frac{2\cdot (i+2)}{2^i}\right) \ .
$
Hence our probability of failure can be upper bounded by:
\[
\Pr{\mathsf{fail}} \leq \sum_{i=0}^{\log k}e^{-i-2} \leq \frac{1}{e^2} \frac{1}{1-1/e}\leq \frac{1}{4} \ .
\]
\noindent\textsf{Claim 3:} With probability at least $3/4$, for any set $K$, with $|K| \leq k$, for $k > n/\alpha$, some query $Q \in \mathcal{Q}_2$ satisfies $|Q \cap K| = 1$.

Similarly as in claim $2$ take any query $Q  \in \mathcal{Q}_2$ and set $K$. Let $k^* = |K|$. We know that $k^* \leq k$. 
We have:
\[
\Pr{|Q \cap K| = 1} = k^* \cdot \frac{1}{6k} \cdot \left(1-\frac{1}{6k}\right)^{k^*-1} \geq  \frac{k^*}{6k} \cdot \left(1-\frac{k^*}{6k}\right) \geq \frac{k^*}{8k} \ .
\]
Hence, if $k^* \in [2^i,2^{i+1}]$, then $\Pr{|Q \cap K| = 1}  \geq \frac{2^i}{8k}$. We denote event $\mathsf{fail}$ as the event that $\mathcal{Q}_2$ fails to hit any set with a most $k$ elements for $k > n/\alpha$:
\[
\Pr{\mathsf{fail}} \leq \sum_{i=0}^{\log k} 2^{i+1} {n \choose 2^{i+1}} \left(1 - \frac{2^i}{8k}\right)^{t_2} \leq \sum_{i=0}^{\log k} e^{i+1} \cdot e^{2^i \ln \frac{ne}{2^i}} e^{-t_2 \cdot \frac{2^i }{8k}} \ .
\]
Knowing that $t_2> 8k(\ln(ne)+ 4)$, we have for any $i\geq 0$: $t_2 > 8k\left(\ln\frac{ne}{2^i} + \frac{2\cdot (i+2)}{2^i}\right)$.
Hence our probability of failure can be upper bounded by:
\[
\Pr{\mathsf{fail}} \leq \sum_{i=0}^{\log k}e^{-i-2} \leq \frac{1}{e^2} \frac{1}{1-1/e}\leq \frac{1}{4} \ .
\]

The probability that any claim fails is at most $1/3 + 1/4 + 1/4 < 1$. By the probabilistic method we have that a sequence satisfying all three claims exist. Now if we want to distinguish $K_1$ from $K_2$ we take $K = K_1 \symdiff K_2$ and observe that by Claim 2 and 3, some query $Q$ has intersection of size exactly $1$ with $K_1 \symdiff K_2$. By Claim 1, each query has at most $\alpha$ elements hence feedback from query $Q$ under $\cF_\alpha$ will be different for $K_1$ and $K_2$. Hence, $\mathcal{Q}$ provides different feedbacks for any two sets of at most $k$ elements.

If $t < n$, then surely $\alpha \geq 2 \log_2 n$ and we use the sequence of queries $\mathcal{Q}$. Otherwise we simply pick a Round-Robin selector of size $n$, where each query contains one unique element. Hence the final query complexity is $\min\{t,n\}$.
\end{proof}

%% file: future.tex

\section{Discussion of results and open directions}
\label{sec:future}
Considering only polynomially-constructible query systems leaves some interesting open directions. One such open direction is whether optimal-length query sequence can be constructed in polynomial time or perhaps it is possible to show some reduction that constructing a close-to-minimum query sequence is hard (even if we know that it exists). Shrinking polylogarithmic gaps between lower and upper bounds (existential) is another challenging direction, as well as considering other interesting classes of feedback with an $\alpha$-capped feedback, e.g., parity. We also believe that with some adjustment, \gt codes could be applied to efficiently solve many open problems in online streaming and graph learning fields.

%% file: main.bbl
\begin{thebibliography}{10}

\bibitem{bay2020optimal}
W.~H. Bay, E.~Price, and J.~Scarlett.
\newblock Optimal non-adaptive probabilistic group testing in general sparsity
  regimes.
\newblock {\em arXiv preprint arXiv:2006.01325}, 2020.

\bibitem{Bshouty09}
N.~H. Bshouty.
\newblock Optimal algorithms for the coin weighing problem with a spring scale.
\newblock In {\em {COLT} 2009 - The 22nd Conference on Learning Theory,
  Montreal, Quebec, Canada, June 18-21, 2009}, 2009.

\bibitem{cao2014quantitative}
C.-C. Cao, C.~Li, and X.~Sun.
\newblock Quantitative group testing-based overlapping pool sequencing to
  identify rare variant carriers.
\newblock {\em BMC bioinformatics}, 15(1):1--14, 2014.

\bibitem{capetanakis1979generalized}
J.~Capetanakis.
\newblock Generalized tdma: The multi-accessing tree protocol.
\newblock {\em IEEE Transactions on Communications}, 27(10):1476--1484, 1979.

\bibitem{capetanakis1979tree}
J.~Capetanakis.
\newblock Tree algorithms for packet broadcast channels.
\newblock {\em IEEE transactions on information theory}, 25(5):505--515, 1979.

\bibitem{Censor-HillelHL15}
K.~{Censor{-}Hillel}, B.~Haeupler, N.~A. Lynch, and M.~M{\'{e}}dard.
\newblock Bounded-contention coding for the additive network model.
\newblock {\em Distributed Comput.}, 28(5):297--308, 2015.

\bibitem{cheraghchi2019simple}
M.~Cheraghchi and J.~Ribeiro.
\newblock Simple codes and sparse recovery with fast decoding.
\newblock In {\em 2019 IEEE International Symposium on Information Theory
  (ISIT)}, pages 156--160. IEEE, 2019.

\bibitem{ChlebusK05}
B.~S. Chlebus and D.~R. Kowalski.
\newblock Almost optimal explicit selectors.
\newblock In M.~Liskiewicz and R.~Reischuk, editors, {\em Fundamentals of
  Computation Theory, 15th International Symposium, {FCT} 2005, L{\"{u}}beck,
  Germany, August 17-20, 2005, Proceedings}, volume 3623 of {\em Lecture Notes
  in Computer Science}, pages 270--280. Springer, 2005.

\bibitem{choi2010optimal}
S.-S. Choi and J.~H. Kim.
\newblock Optimal query complexity bounds for finding graphs.
\newblock {\em Artificial Intelligence}, 174(9-10):551--569, 2010.

\bibitem{ClementiMS01}
A.~E.~F. Clementi, A.~Monti, and R.~Silvestri.
\newblock Selective families, superimposed codes, and broadcasting on unknown
  radio networks.
\newblock In {\em Proceedings of the Twelfth Annual Symposium on Discrete
  Algorithms, January 7-9, 2001, Washington, DC, {USA}}, pages 709--718.
  {ACM/SIAM}, 2001.

\bibitem{clifford2010pattern}
R.~Clifford, K.~Efremenko, E.~Porat, and A.~Rothschild.
\newblock Pattern matching with don't cares and few errors.
\newblock {\em Journal of Computer and System Sciences}, 76(2):115--124, 2010.

\bibitem{Coja-OghlanGHL20}
A.~Coja{-}Oghlan, O.~Gebhard, M.~Hahn{-}Klimroth, and P.~Loick.
\newblock Optimal group testing.
\newblock In J.~D. Abernethy and S.~Agarwal, editors, {\em Conference on
  Learning Theory, {COLT} 2020, 9-12 July 2020, Virtual Event [Graz, Austria]},
  volume 125 of {\em Proceedings of Machine Learning Research}, pages
  1374--1388. {PMLR}, 2020.

\bibitem{cormode2005holistic}
G.~Cormode, M.~Garofalakis, S.~Muthukrishnan, and R.~Rastogi.
\newblock Holistic aggregates in a networked world: Distributed tracking of
  approximate quantiles.
\newblock In {\em Proceedings of the 2005 ACM SIGMOD international conference
  on Management of data}, pages 25--36, 2005.

\bibitem{cormode2008finding}
G.~Cormode and M.~Hadjieleftheriou.
\newblock Finding frequent items in data streams.
\newblock {\em Proceedings of the VLDB Endowment}, 1(2):1530--1541, 2008.

\bibitem{cormode2003finding}
G.~Cormode, F.~Korn, S.~Muthukrishnan, and D.~Srivastava.
\newblock Finding hierarchical heavy hitters in data streams.
\newblock In {\em Proceedings 2003 VLDB Conference}, pages 464--475. Elsevier,
  2003.

\bibitem{cormode2005s}
G.~Cormode and S.~Muthukrishnan.
\newblock What's hot and what's not: tracking most frequent items dynamically.
\newblock {\em ACM Transactions on Database Systems (TODS)}, 30(1):249--278,
  2005.

\bibitem{cormode2006combinatorial}
G.~Cormode and S.~Muthukrishnan.
\newblock Combinatorial algorithms for compressed sensing.
\newblock In {\em International colloquium on structural information and
  communication complexity}, pages 280--294. Springer, 2006.

\bibitem{DAMASCHKE}
P.~Damaschke.
\newblock Threshold group testing.
\newblock {\em Electronic Notes in Discrete Mathematics}, 21:265 -- 271, 2005.
\newblock General Theory of Information Transfer and Combinatorics.

\bibitem{BonisGV03}
A.~{De Bonis}, L.~Gasieniec, and U.~Vaccaro.
\newblock Generalized framework for selectors with applications in optimal
  group testing.
\newblock In {\em Automata, Languages and Programming, 30th International
  Colloquium, {ICALP} 2003, Eindhoven, The Netherlands, June 30 - July 4, 2003.
  Proceedings}, volume 2719 of {\em Lecture Notes in Computer Science}, pages
  81--96. Springer, 2003.

\bibitem{MarcoJK19}
G.~{De Marco}, T.~Jurdzinski, and D.~R. Kowalski.
\newblock Optimal channel utilization with limited feedback.
\newblock In L.~A. Gasieniec, J.~Jansson, and C.~Levcopoulos, editors, {\em
  Fundamentals of Computation Theory - 22nd International Symposium, {FCT}
  2019, Copenhagen, Denmark, August 12-14, 2019, Proceedings}, volume 11651 of
  {\em Lecture Notes in Computer Science}, pages 140--152. Springer, 2019.

\bibitem{de2013searching}
G.~De~Marco and D.~R. Kowalski.
\newblock Searching for a subset of counterfeit coins: Randomization vs
  determinism and adaptiveness vs non-adaptiveness.
\newblock {\em Random Structures \& Algorithms}, 42(1):97--109, 2013.

\bibitem{djackov1975search}
A.~Djackov.
\newblock On a search model of false coins.
\newblock In {\em Topics in Information Theory (Colloquia Mathematica
  Societatis Janos Bolyai 16). Budapest, Hungary: Hungarian Acad. Sci}, pages
  163--170, 1975.

\bibitem{SURV}
B.~Doerr.
\newblock Probabilistic tools for the analysis of randomized optimization
  heuristics.
\newblock {\em CoRR}, abs/1801.06733, 2018.

\bibitem{dorfman1943detection}
R.~Dorfman.
\newblock The detection of defective members of large populations.
\newblock {\em The Annals of Mathematical Statistics}, 14(4):436--440, 1943.

\bibitem{duhwang}
D.~Du, F.~K. Hwang, and F.~Hwang.
\newblock {\em Combinatorial group testing and its applications}, volume~12.
\newblock World Scientific, 2000.

\bibitem{Feige20}
U.~Feige and A.~Lellouche.
\newblock Quantitative group testing and the rank of random matrices.
\newblock {\em CoRR}, abs/2006.09074, 2020.

\bibitem{GalanisWJD03}
L.~Galanis, Y.~Wang, S.~R. Jeffery, and D.~J. DeWitt.
\newblock Locating data sources in large distributed systems.
\newblock In J.~C. Freytag, P.~C. Lockemann, S.~Abiteboul, M.~J. Carey, P.~G.
  Selinger, and A.~Heuer, editors, {\em Proceedings of 29th International
  Conference on Very Large Data Bases, {VLDB} 2003, Berlin, Germany, September
  9-12, 2003}, pages 874--885. Morgan Kaufmann, 2003.

\bibitem{gallager1985perspective}
R.~Gallager.
\newblock A perspective on multiaccess channels.
\newblock {\em IEEE Transactions on information Theory}, 31(2):124--142, 1985.

\bibitem{Gebhard19}
O.~Gebhard, M.~Hahn{-}Klimroth, D.~Kaaser, and P.~Loick.
\newblock Quantitative group testing in the sublinear regime.
\newblock {\em CoRR}, abs/1905.01458, 2019.

\bibitem{gibbons2002fast}
P.~B. Gibbons, Y.~Matias, and V.~Poosala.
\newblock Fast incremental maintenance of approximate histograms.
\newblock {\em ACM Transactions on Database Systems (TODS)}, 27(3):261--298,
  2002.

\bibitem{gilbert2002fast}
A.~C. Gilbert, S.~Guha, P.~Indyk, Y.~Kotidis, S.~Muthukrishnan, and M.~J.
  Strauss.
\newblock Fast, small-space algorithms for approximate histogram maintenance.
\newblock In {\em Proceedings of the thiry-fourth annual ACM symposium on
  Theory of computing}, pages 389--398, 2002.

\bibitem{gilbert2002summarize}
A.~C. Gilbert, Y.~Kotidis, S.~Muthukrishnan, and M.~J. Strauss.
\newblock How to summarize the universe: Dynamic maintenance of quantiles.
\newblock In {\em VLDB'02: Proceedings of the 28th International Conference on
  Very Large Databases}, pages 454--465. Elsevier, 2002.

\bibitem{GrebinskiK00}
V.~Grebinski and G.~Kucherov.
\newblock Optimal reconstruction of graphs under the additive model.
\newblock {\em Algorithmica}, 28(1):104--124, 2000.

\bibitem{greenberg1987estimating}
A.~G. Greenberg, P.~Flajolet, and R.~E. Ladner.
\newblock Estimating the multiplicities of conflicts to speed their resolution
  in multiple access channels.
\newblock {\em Journal of the ACM (JACM)}, 34(2):289--325, 1987.

\bibitem{greenberg1985lower}
A.~G. Greenberg and S.~Winograd.
\newblock A lower bound on the time needed in the worst case to resolve
  conflicts deterministically in multiple access channels.
\newblock {\em Journal of the ACM (JACM)}, 32(3):589--596, 1985.

\bibitem{greenwald2001space}
M.~Greenwald and S.~Khanna.
\newblock Space-efficient online computation of quantile summaries.
\newblock {\em ACM SIGMOD Record}, 30(2):58--66, 2001.

\bibitem{IlyasAE04}
I.~F. Ilyas, W.~G. Aref, and A.~K. Elmagarmid.
\newblock Supporting top-k join queries in relational databases.
\newblock {\em {VLDB} J.}, 13(3):207--221, 2004.

\bibitem{Indyk97}
P.~Indyk.
\newblock Deterministic superimposed coding with applications to pattern
  matching.
\newblock In {\em 38th Annual Symposium on Foundations of Computer Science,
  {FOCS} '97, Miami Beach, Florida, USA, October 19-22, 1997}, pages 127--136.
  {IEEE} Computer Society, 1997.

\bibitem{Indyk02}
P.~Indyk.
\newblock Explicit constructions of selectors and related combinatorial
  structures, with applications.
\newblock In D.~Eppstein, editor, {\em Proceedings of the Thirteenth Annual
  {ACM-SIAM} Symposium on Discrete Algorithms, January 6-8, 2002, San
  Francisco, CA, {USA}}, pages 697--704. {ACM/SIAM}, 2002.

\bibitem{IndykNR10}
P.~Indyk, H.~Q. Ngo, and A.~Rudra.
\newblock Efficiently decodable non-adaptive group testing.
\newblock In M.~Charikar, editor, {\em Proceedings of the Twenty-First Annual
  {ACM-SIAM} Symposium on Discrete Algorithms, {SODA} 2010, Austin, Texas, USA,
  January 17-19, 2010}, pages 1126--1142. {SIAM}, 2010.

\bibitem{kautz1964nonrandom}
W.~Kautz and R.~Singleton.
\newblock Nonrandom binary superimposed codes.
\newblock {\em IEEE Transactions on Information Theory}, 10(4):363--377, 1964.

\bibitem{KomlosG85}
J.~Koml{\'{o}}s and A.~G. Greenberg.
\newblock An asymptotically fast nonadaptive algorithm for conflict resolution
  in multiple-access channels.
\newblock {\em {IEEE} Trans. Inf. Theory}, 31(2):302--306, 1985.

\bibitem{lee2019saffron}
K.~Lee, K.~Chandrasekher, R.~Pedarsani, and K.~Ramchandran.
\newblock Saffron: A fast, efficient, and robust framework for group testing
  based on sparse-graph codes.
\newblock {\em IEEE Transactions on Signal Processing}, 67(17):4649--4664,
  2019.

\bibitem{massey1981collision}
J.~L. Massey.
\newblock Collision-resolution algorithms and random-access communications.
\newblock In {\em Multi-user communication systems}, pages 73--137. Springer,
  1981.

\bibitem{PR11}
E.~Porat and A.~Rothschild.
\newblock Explicit nonadaptive combinatorial group testing schemes.
\newblock {\em {IEEE} Trans. Inf. Theory}, 57(12):7982--7989, 2011.

\bibitem{TUZ}
A.~Ta{-}Shma, C.~Umans, and D.~Zuckerman.
\newblock Loss-less condensers, unbalanced expanders, and extractors.
\newblock In J.~S. Vitter, P.~G. Spirakis, and M.~Yannakakis, editors, {\em
  Proceedings on 33rd Annual {ACM} Symposium on Theory of Computing, July 6-8,
  2001, Heraklion, Crete, Greece}, pages 143--152. {ACM}, 2001.

\bibitem{yu2004false}
J.~X. Yu, Z.~Chong, H.~Lu, and A.~Zhou.
\newblock False positive or false negative: Mining frequent itemsets from high
  speed transactional data streams.
\newblock In {\em VLDB}, volume~4, pages 204--215, 2004.

\end{thebibliography}
